\documentclass[final]{siamart1116}
\usepackage{tikz}
\usepackage{pgfplots}

\def\E{\mathbb{E}}
\def\P{\mathbb{P}}
\def\R{\mathbb{R}}

\def\VV{\mathcal{V}}
\newcommand{\Markov}[2]{\underset{#1}{\overset{#2}{\rightleftharpoons}}}
\newcommand{\kstep}{k_{\text{step}}}
\newcommand{\kon}{k_{\text{on}}}
\newcommand{\koff}{k_{\text{off}}}
\newcommand{\lstep}{\lambda_{\text{step}}}
\newcommand{\lon}{\lambda_{\text{on}}}
\newcommand{\loff}{\lambda_{\text{off}}}

\newcommand{\pon}{p_{\text{on}}}

\newcommand{\ton}{\tau_{\text{on}}}
\newcommand{\toff}{\tau_{\text{off}}}
\newcommand{\Ron}{R_{\text{on}}}
\newcommand{\Roff}{R_{\text{off}}}
\newcommand{\Ton}{T_{\text{on}}}
\newcommand{\Toff}{T_{\text{off}}}

\newcommand{\Yon}{Y_{\text{on}}}
\newcommand{\Yoff}{Y_{\text{off}}}

\newcommand{\J}{\textbf{J}}

\newcommand{\ol}{\overline}
\newcommand{\x}{\overline{x}}
\newcommand{\z}{\overline{z}}

\newcommand{\dd}{\text{d}}

\newcommand{\eps}{\varepsilon}

\renewcommand{\P}{\mathbb{P}}

\newcommand{\expect}{\mathbb{E}}

\theoremstyle{plain}

\theoremstyle{remark}

\theoremstyle{conjecture}

\newtheorem{prop}[theorem]{Proposition}

\usepackage{lipsum}
\usepackage{amsfonts}
\usepackage{graphicx}
\usepackage{epstopdf}
\usepackage{algorithmic}
\ifpdf
  \DeclareGraphicsExtensions{.eps,.pdf,.png,.jpg}
\else
  \DeclareGraphicsExtensions{.eps}
\fi

\newcommand{\TheTitle}{Analysis of non-processive molecular motor transport using renewal reward theory}
\newcommand{\ShortTitle}{Non-processive motor transport}
\newcommand{\TheAuthors}{Christopher E. Miles, Sean D. Lawley, and James P. Keener}

\headers{\ShortTitle}{\TheAuthors}

\title{{\TheTitle}\thanks{Submitted to the editors \today.
\funding{This work was supported by NSF grant DMS-RTG 1148230. CEM and JPK were also supported by NSF grant DMS 1515130. SDL was also supported by NSF grant DMS 1814832.}}
}

\author{Christopher E. Miles\thanks{Department of Mathematics, University of Utah, Salt Lake City, UT 84112 USA ({\email{miles@math.utah.edu}}).}
\and Sean D. Lawley\thanks{Department of Mathematics, University of Utah, Salt Lake City, UT 84112 USA ({\email{ lawley@math.utah.edu}}).}
\and James P. Keener\thanks{Departments of Mathematics and Bioengineering, University of Utah, Salt Lake City, UT 84112 USA ({\email{keener@math.utah.edu}}).}}  

\usepackage{amsopn}

\ifpdf
\hypersetup{
  pdftitle={\TheTitle},
  pdfauthor={\TheAuthors}
}
\fi

\begin{document}

\maketitle

\begin{abstract}
We propose and analyze a mathematical model of cargo transport by non-processive molecular motors. In our model, the motors change states by random discrete events (corresponding to stepping and binding/unbinding), while the cargo position follows a stochastic differential equation (SDE) that depends on the discrete states of the motors. The resulting system for the cargo position is consequently an SDE that randomly switches according to a Markov jump process governing motor dynamics. To study this system we (1) cast the cargo position in a renewal theory framework and generalize the renewal reward theorem and (2) decompose the continuous and discrete sources of stochasticity and exploit a resulting pair of disparate timescales. With these mathematical tools, we obtain explicit formulas for experimentally measurable quantities, such as cargo velocity and run length. Analyzing these formulas then yields some predictions regarding so-called non-processive clustering, the phenomenon that a single motor cannot transport cargo, but that two or more motors can. We find that having motor stepping, binding, and unbinding rates depend on the number of bound motors, due to geometric effects, is necessary and sufficient to explain recent experimental data on non-processive motors.
\end{abstract}

\begin{keywords}
molecular motors, intracellular transport, reward renewal theory, stochastic hybrid systems, switching SDEs
\end{keywords}
\begin{AMS}
92C05, 
60K20, 
60J28, 
34F05, 
60H10 
\end{AMS}

\section{Introduction}

Active intracellular transport of cargo (such as organelles) is critical to cellular function. The primary type of active transport involves molecular motors, which alternate between epochs of active transport (discrete stepping) along a microtubule and epochs of passive diffusion when the motors are unbound from the microtubule. Stochastic modeling of this fundamental process has a rich and fruitful history (see the review \cite{Bressloff2013}).

In both experimental and modeling studies, \emph{processive} motors have received considerable attention. Processive motors are characterized by taking hundreds of steps along a microtubule before unbinding. In contrast, \emph{non-processive} motors (such as most members of the kinesin-14 family) take very few (1 to 5) steps before unbinding from a microtubule \cite{Case1997,Endow2000}. Non-processive motors are crucial to a number of cellular processes, including directing cytoskeletal filaments \cite{Shaklee2008}, driving microtubule-microtubule sliding during mitosis \cite{Fink2009}, and retrograde transport along microtubules in plants \cite{Yamada2017}. Here, we focus on motor behavior during transport.

Some curious properties of non-processive motor transport were found in \cite{Furuta2013}. \emph{One} non-processive (Ncd) motor has extremely limited transport ability, measured by both velocity and run length (distance traveled before detaching from a microtubule). However, \textit{two} non-processive motors somehow act in unison to produce significant directed motion, a phenomenon termed ``clustering.'' This observation is supported by the subsequent studies \cite{Jonsson2015,Mieck2015}, where similar experiments were performed creating a mutant of attached non-processive kinesin-14 motors, and processivity emerges. Moreover, the authors of \cite{Furuta2013} note that adding more Ncd motors beyond two further increases transport ability. In contrast, one processive motor (kinesin-1) is sufficient to produce transport, and additional motors do not significantly increase transport ability \cite{Furuta2013,Shubeita2008}. Other interesting facets of transport by non-processive motors include the emergence of processive transport in the presence of higher microtubule concentration \cite{Furuta2008} or opposing motors \cite{Hodges2009}.

In this work, we formulate and analyze a mathematical model to investigate the natural question: {how do non-processive motors cooperate to transport cargo?} Our model predicts that non-processive motor stepping, binding, and unbinding rates must depend on the number of bound motors, and that this dependence is a key mechanism driving the collective transport of non-processive motors. We note that such dependence has been observed in experiments \cite{Feng2017,Grotjahn2017} and in simulations of detailed computational models \cite{Korn2009, Erickson2011,Lombardo2017,Furuta2008}, all stemming from geometric effects of cargo/motor configuration. 

Non-processive motors are notoriously difficult to study experimentally, because they take only a few steps before detaching. For this same reason, it is not clear how to best model non-processive motors, or known if existing modeling frameworks, such as mean-field methods \cite{Julicher1995,Duke2000a} or averaging the stepping dynamics into an effective velocity \cite{McKinley2011a}, are appropriate. Hence, our model explicitly includes the discrete binding, unbinding, and stepping dynamics of each motor, as well as the continuous tethered motion of the cargo.

Mathematically, our model takes the form of a randomly switching stochastic differential equation (SDE), and thus merges continuous dynamics with discrete events. The continuous SDE dynamics track the cargo position, while the discrete events correspond to motor binding, unbinding, and stepping. Our model is thus a \emph{stochastic hybrid system} \cite{Bressloff2014}, which are often two-component processes, $(J(t),X(t))_{t\ge0}\in\mathcal{I}\times\R^{d}$, where $J$ is a Markov jump process on a finite set $\mathcal{I}$, and $X$ evolves continuously by
\begin{align}\label{introeq}
\dd X(t)=F_{J(t)}(X(t))\, \dd t+\sigma\, \dd W(t),
\end{align}
where $\{F_{j}(x)\}_{j\in\mathcal{I}}$ is a given finite family of vector fields, $\sigma\ge0$, and $W$ is a Brownian motion. That is, $X$ follows an SDE whose righthand side switches according to the process $J$.

However, our model differs from most previous hybrid systems in some key ways. First, the set of possible continuous dynamics (e.g.\ possible righthand sides of (\ref{introeq})) for our model is infinite. Second, the new righthand side of (\ref{introeq}) that is chosen when $J$ jumps depends on the value of $X$ at that jump time, although the rates dictating $J$ are taken to be independent of $X$. 

We employ several techniques to analyze our model and make predictions regarding non-processive motor transport. First, we cast our model in a renewal theory framework, and generalize the classical renewal reward theorem \cite{Vlasiou2011} to apply to our setting, distinct from previous motor applications \cite{Krishnan2011,Hughes2011,Hughes2012,Hughes2013,Shtylla2015}. Next, we decompose the stochasticity in the system by averaging over the diffusion while conditioning on a realization of the jump process. This effectively turns the randomly switching SDE into a randomly switching ordinary differential equation (ODE), and thus a piecewise deterministic Markov process \cite{Davis1984}. Finally, we observe that for biologically reasonable parameter values, the relaxation rate of the continuous cargo dynamics is much faster than the jump rates for the discrete motor behavior. We then exploit this timescale separation to find explicit formulas for key motor transport statistics. 

The rest of the paper is organized as follows. We formulate the mathematical model in section~\ref{section setup}. In section~\ref{section renewal}, we generalize the renewal reward theorem to apply to our model. In section~\ref{section analysis}, we derive explicit formulas to evaluate motor transport. In section~\ref{section biology}, we use the model to make biological predictions. We conclude with a brief discussion and an Appendix that collects several proofs.

\section{Mathematical model}\label{section setup}

We model the motion of a single cargo driven by $M\ge1$ motors along a single microtubule.  These motors are permanently attached to the cargo, but they can bind to and unbind from the microtubule. At any time $t\ge0$, the state of our model is specified by
\begin{align*}
\big(X(t),\textbf{Z}(t),\J(t)\big)\in\R\times\R^{M}\times\{u,b\}^{M},
\end{align*}
where $X(t)\in\R$ is the location of the center of the cargo, $\textbf{Z}(t)=(Z_{i}(t))_{i=1}^{M}\in\R^{M}$ gives the locations of the centers of $M$ motors, and $\J(t)=(J_{i}(t))_{i=1}^{M}$ specifies if each motor is unbound or bound. 
Spatial locations are measured along the principal axis of the microtubule, which we identify with the real line. 

The cargo position evolves continuously in time, while the positions and states of motors change by discrete events, which correspond to binding to the microtubule, stepping along the microtubule, or unbinding from the microtubule. Specifically, in between these discrete motor events, $X(t)$ follows an Ornstein-Uhlenbeck (OU) process centered at the average bound motor position,
\begin{align}\label{ou}
\dd X(t)=\frac{k}{\gamma}\sum_{i\in I(t)}\big(Z_{i}(t)-X(t)\big) \dd t+\sqrt{2k_{B}T/\gamma}\,\dd W(t).
\end{align}
Here, $I(t)=\{i:J_{i}(t)=b\}\subseteq \{1,\dots,M\}$ gives the indices of motors that are bound at time $t\ge0$, and $\{W(t)\}_{t\ge0}$ is a standard Brownian motion. The SDE (\ref{ou}) stems from assuming a viscous (low Reynolds number) regime with drag coefficient $\gamma>0$, and that each bound motor exerts a Hookean force on the cargo with stiffness $k>0$. The Stokes-Einstein relation specifies the diffusion coefficient $k_{B}T/\gamma$, where $k_{B}$ is Boltzmann's constant and $T$ is the absolute temperature.

\begin{figure}[t]
\begin{centering}
\includegraphics[width=.4\textwidth]{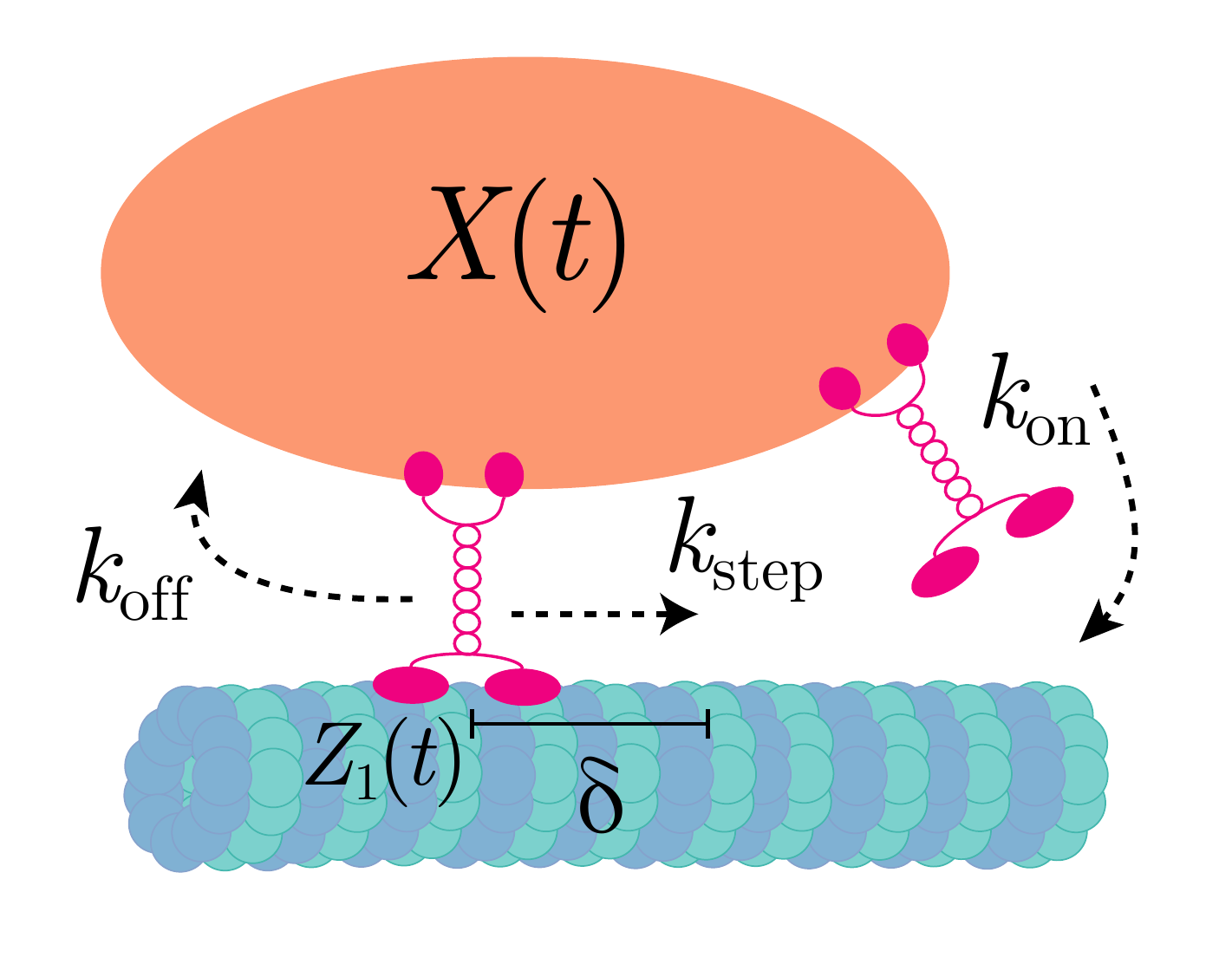}
\caption{\label{fig:fig1} 
Schematic describing the binding, unbinding, and stepping of motors. The positions of cargo and bound motors are $X(t)$ and $Z_i(t)$, respectively, both measured with respect to the principal axis of the microtubule. The state of the motor can switch between bound or unbound, and while bound, the motor can step, incrementing $Z_i(t)$ by displacement $\delta$. 
}
\end{centering}
\end{figure}

The discrete behavior of motors is as follows. Let $m(t)\in\{0,1,\dots,M\}$ denote the number of bound motors at time $t\ge0$,
\begin{align*}
m(t)=\sum_{i=1}^{M}1_{\{J_{i}(t)\neq u\}}\in\{0,1,\dots,M\},
\end{align*}
where $1_{\{A\}}$ denotes the indicator function on an event $A$. Each unbound motor binds to the microtubule at rate $\kon(m(t))>0$. Since unbound motors are tethered to the cargo, if an unbound motor binds at time $t\ge0$, then we assume that it binds to the track at $X(t)$ (motors can bind anywhere along the microtubule, not only binding sites). We could allow it to bind to a random position, but if the mean binding position is $X(t)$, then our results are unchanged. The position of each bound motor is fixed until it either steps or unbinds. Each bound motor unbinds at rate $\koff(m(t))>0$ and steps at rate $\kstep(m(t))>0$. When a motor steps, we add $\delta>0$ to its position, which is then fixed until it steps again or unbinds. This discrete motor behavior is summarized in Fig.~\ref{fig:fig1}. We emphasize that the motor binding, unbinding, and stepping rates are allowed to depend on the number of bound motors, $m(t)$, but are otherwise independent of $X(t)$.

\subsection{Nondimensionalization and assumptions}\label{section setup2}

We now give a dimensionless and more precise formulation of the model described above. First, we nondimensionalize the model by rescaling time by the rate $\koff(1)$ and space by the inverse length $\delta^{-1}$. Next, we note that unbound motors do not affect the cargo position. Hence, for convenience we can take $Z_{i}(t)=X(t)$ if the $i$-th motor is unbound, meaning that we can include unbound motors in the sum in (\ref{ou}) with zero contribution, and make the sum over \textit{all} motors. This yields the simplified dimensionless form
\begin{align}\label{dim}
\dd X(t)&=\eps^{-1}\sum_{i=1}^{M}\big(Z_{i}(t)-X(t)\big)\dd t+\sigma \dd W(t),
\end{align}
where
\begin{align*}
\eps:=\koff(1)\gamma/k,\qquad
\sigma:=\sqrt{2k_{B}T/(\delta^{2}\koff(1)\gamma)},
\end{align*}
and motors bind, unbind, and step at dimensionless rates
\begin{align}
\lon(m):=\frac{\kon(m)}{\koff(1)},\quad
\loff(m):=\frac{\koff(m)}{\koff(1)},\quad
\lstep(m):=\frac{\kstep(m)}{\koff(1)}. \label{eq:rates}
\end{align}

We find it convenient for our analysis to track the number of steps taken by each motor before unbinding, so let us expand the state space of $\J(t)$ so that its components $(J_{i}(t))_{i=1}^{M}$ each take values in $\{u,0,1,2,\dots\}$ with transition rates
\begin{align}\label{jumprates}
u\overset{\lon(m(t))}{\to}0,\qquad
j\overset{\lstep(m(t))}{\to}j+1,
\qquad
j\overset{\loff(m(t))}{\to}u,\qquad j\neq u.
\end{align}
The components of $\J(t)$ are conditionally independent given $m(t)$. At time $t\ge0$, the $i$-th motor is unbound if $J_{i}(t)=u$, bound if $J_{i}(t)\ge0$, and steps when $J_{i}(t)$ transitions from $j$ to $j+1$ for $j\ge0$. 
Under these assumptions, $m(t)$ is itself a Markov process on $\{0,1,\dots,M\}$ with transition rates
\begin{equation}
	0 \Markov{1}{M\lon(0)} 1 \Markov{2\loff(2)}{(M-1)\lon(1)} 2\Markov{}{}
	  \cdots \Markov{}{}M-2\Markov{(M-1)\loff(M-1)}{2\lon(M-2)} M-1 \Markov{M\loff(M)}{\lon(M-1)} M. \label{eq:nMarkov}
\end{equation}

For simplicity, we assume that the motors are initially unbound and that the cargo and motors start at the origin,
\begin{align*}
J_{i}(0)=u,\quad X(0)=Z_{i}(0)=0,\quad i\in\{1,\dots,M\}.
\end{align*}
The position of the $i$-th motor is then
\begin{align}\label{zp}
Z_{i}(t)=\big(X(\tau_{i}(t))+J_{i}(t)\big)1_{\{J_{i}(t)\neq u\}}+X(t)1_{\{J_{i}(t)=u\}},\quad i\in\{1,\dots,M\},
\end{align}
where $\tau_{i}(t)$ is the most recent binding time of the $i$-th motor,
\begin{align*}
\tau_{i}(t)=\sup\{s<t:J_{i}(s)=u\},\quad i\in\{1,\dots,M\}.
\end{align*}
We assume the Brownian motion $W=\{W(t)\}_{t\ge0}$ and the jump process $\J=\{\J(t)\}_{t\ge0}$ are independent. 

\section{Cargo position as a renewal reward process}\label{section renewal}

In order to analyze our model, we first show that $X(t)$ is a renewal reward process with partial rewards \cite{Vlasiou2011} and extend the classical renewal reward theorem to our case of partial rewards. This framework has an intuitive interpretation: the net displacement of cargo is determined by the displacement accrued at each epoch of being bound or unbound. However, there is a technical challenge. In the most classical setting, the reward-renewal theorem accrues rewards at the end of each epoch and boundedness of expectation of the rewards is sufficient to apply the reward-renewal theorem. In the case of partial rewards (which we have in our model), where rewards are accrued during an epoch, stronger conditions are required, which we prove are satisfied. 

First, define the sequence of times in which the cargo completely detaches from the microtubule (off) and subsequently reattaches to the microtubule (on),
\begin{align*}
0=\ton^{0}=\toff^{0}<\ton^{1}<\toff^{1}<\ton^{2}<\toff^{2}<\dots
\end{align*}
by
\begin{align}\label{tauonoff}
\begin{split}
\toff^{k}&:=\inf\{t>\ton^{k}:m(t)=0\},\quad k\ge1,\\
\ton^{k}&:=\inf\{t>\toff^{k-1}:m(t)\ge1\},\quad k\ge1.
\end{split}
\end{align}
Next, define the sequence of cargo displacements when the cargo is attached to the microtubule (on) and detached from the microtubule (off),
\begin{align}\label{RonRoff}
\Ron^{k}&:=X(\toff^{k})-X(\ton^{k}),
\quad
\Roff^{k}:=X(\ton^{k})-X(\toff^{k-1}),
\quad k\ge1,
\end{align}
and the corresponding times spent attached or detached,
\begin{align}\label{TonToff}
\Ton^{k}&:=\toff^{k}-\ton^{k},
\quad
\Toff^{k}:=\ton^{k}-\toff^{k-1},
\quad k\ge1.
\end{align}
It follows directly from the strong Markov property that $\{(\Toff^{k}+\Ton^{k},\Roff^{k}+\Ron^{k})\}_{k\ge1}$ is an independent and identically distributed (iid) sequence of random variables. 

In the language of renewal theory, $\{\Toff^{k}+\Ton^{k}\}_{k\ge1}$ are the interarrival times and $\{\Roff^{k}+\Ron^{k}\}_{k\ge1}$ are the corresponding rewards. Let $N(t)$ be the renewal process that counts the number of arrivals before time $t\ge0$,
\begin{align}\label{Ndefn}
N(t):=\sup\{k\ge0:\toff^{k}\le t\}.
\end{align}
Define the reward function, $R(t)$, and the partial reward function, $Y(t)$, by
\begin{align}\label{Y}
R(t):=\sum_{k=1}^{N(t)}(\Ron^{k}+\Roff^{k}),
\qquad
Y(t):=X(t)-X(\toff^{N(t)}),
\end{align}
and observe that
\begin{align*}
X(t)=R(t)+Y(t).
\end{align*}
In words, $R(t)$ describes rewards accrued during past epochs, and $Y(t)$ is the partial reward accrued during the current epoch. We show below that $\E[|\Ron+\Roff|]<\infty$ and $\E[\Ton+\Toff]<\infty$, and therefore the classical renewal reward theorem \cite{Vlasiou2011} ensures that
\begin{align}\label{rr}
\lim_{t\to\infty}\frac{R(t)}{t}
=\lim_{t\to\infty}\frac{\E[R(t)]}{t}
=\frac{\E[\Ron]+\E[\Roff]}{\E[\Ton]+\E[\Toff]}\quad\text{almost surely}.
\end{align}
The following theorem verifies that this convergence actually holds for $X(t)$.

\begin{theorem}\label{theorem renewal}
The following limit holds,
\begin{align}\label{Ve}
\lim_{t\to\infty}\frac{X(t)}{t}
=\lim_{t\to\infty}\frac{\E[X(t)]}{t}
=\frac{\E[\Ron]+\E[\Roff]}{\E[\Ton]+\E[\Toff]}\quad\text{almost surely}.
\end{align}
\end{theorem}

To prove this theorem, we need several lemmas. We collect the proofs of these lemmas in {Appendix~\ref{section lemma proofs}}. The first lemma bounds the probability that the partial reward function $Y(t)$ in (\ref{Y}) is large when the cargo is detached from the microtubule. 

\begin{lemma}\label{lemma off}
Define the sequence of iid random variables $\{\Yoff^{k}\}_{k\ge1}$ by
\begin{align*}
\Yoff^{k}&:=\sup_{t\in[\toff^{k-1},\ton^{k}]}\big|X(t)-X(\toff^{k})\big|,\quad k\ge1.
\end{align*}
Then for any $C>0$ and $k\ge1$, we have that
\begin{align}\label{bbb}
 \P(\Yoff^{k}\ge C) 
&\le \sqrt{\pi/x}(2x+1)e^{-x},\quad \text{where }x=\tfrac{C}{\sigma}\sqrt{2M\lon(0) }>0.
\end{align}
\end{lemma}

Similarly, the next lemma bounds the probability that the partial reward function is large when the cargo is attached to the microtubule.

\begin{lemma}\label{lemma on}
Define the sequence of iid random variables $\{\Yon^{k}\}_{k\ge1}$ by
\begin{align*}
\Yon^{k}&:=\sup_{t\in[\ton^{k},\toff^{k}]}|X(t)-X(\ton^{k})|,\quad k\ge1.
\end{align*}
\textcolor{black}{There exists $\lambda_{0}>0,\lambda_{1}>0$ so that if $k\ge1$, then
$ \P(\Yon^{k}\ge C) 
\le \lambda_{0}\sqrt{C}e^{-\lambda_{1}\sqrt{C}}$ for all sufficiently large $C>0$.
}
\end{lemma}

The next lemma uses Lemmas~\ref{lemma off} and \ref{lemma on} to prove that the partial reward function gets large only finitely many times.

\begin{lemma}\label{lemma bc}
Define the sequence of iid random variables $\{Y_{k}\}_{k\ge1}$ by
\begin{align}\label{Ydef}
Y_{k}:=\sup_{t\in[\toff^{k-1},\toff^{k}]}\big|X(t)-X(\toff^{k})\big|,\quad k\ge1.
\end{align}
Then
\begin{align*}
\P\Big(\lim_{K\to\infty}\bigcup_{k=K}^{\infty}\big\{Y_{k}>\sqrt{k}\big\}\Big)=0.
\end{align*}
\end{lemma}

The last lemma checks that the mean of $Y_{k}$ in (\ref{Ydef}) is finite.

\begin{lemma}\label{lemma finite}
Define $\{Y_{k}\}_{k\ge1}$ as in (\ref{Ydef}). Then $\E[Y_{k}]<\infty$ for all $k\ge1$.
\end{lemma}

With these lemmas in place, we are ready to prove Theorem~\ref{theorem renewal}.

\begin{proof}[Proof of Theorem~\ref{theorem renewal}]

It follows immediately from Lemma~\ref{lemma finite} that $\E[|\Ron^{k}+\Roff^{k}|]<\infty$. Furthermore, $\Toff^{k}$ is exponentially distributed with rate $M\lon(0)$, and the proof of Lemma~\ref{lemma on} shows that $\E[\Ton^{k}]<\E[S]$ for an exponentially distributed random variable $S$ with some rate $\lambda>0$. Hence, $\E[\Ton^{k}+\Toff^{k}]<\infty$, and thus (\ref{rr}) holds by a direct application of the classical renewal reward theorem \cite{Vlasiou2011}.

Therefore, it remains to check that
\begin{align}\label{double}
\lim_{t\to\infty}\frac{\E\big[X(t)-X(\toff^{N(t)})\big]}{t}
=0=\lim_{t\to\infty}\frac{X(t)-X(\toff^{N(t)})}{t},\quad\text{almost surely}.
\end{align}
The first equality in (\ref{double}) follows immediately from Lemma~\ref{lemma finite}.

To verify the second equality in (\ref{double}), we note that Lemma~\ref{lemma bc} ensures that
\begin{align*}
\limsup_{k\to\infty}\frac{Y_{k}}{\sqrt{k}}\le1,\quad\text{almost surely}.
\end{align*}
Therefore, 
\begin{align*}
\lim_{t\to\infty}\frac{|X(t)-X(\toff^{N(t)})|}{t}
\le\lim_{t\to\infty}\frac{|Y_{N(t)}|}{t}
\le\lim_{t\to\infty}\frac{\sqrt{N(t)}}{t}=0,
\quad\text{almost surely},
\end{align*}
since
\begin{align*}
\lim_{t\to\infty}\frac{N(t)}{t}=\frac{1}{\E[\Ton+\Toff]},
\quad\text{almost surely},
\end{align*}
by the strong law of large numbers for renewal processes \cite{Vlasiou2011}. 
\end{proof}
Consequently, the position of the cargo does indeed satisfy a classical reward-renewal structure with two different types of epochs: bound and unbound, each of which accrue some net displacement. 

\section{Mathematical analysis of transport ability}\label{section analysis}

With the framework of renewal theory constructed in section~\ref{section renewal}, we are ready to analyze the transport ability of the model introduced in section~\ref{section setup}. To assess the transport ability of the motor cargo ensemble, we analyze the expected \emph{run length}, expected \emph{run time}, and asymptotic \emph{velocity}. We define the run length to be the distance traveled by the cargo between the first time a motor attaches to the cargo until the next time that all motors are detached from the microtubule, which was defined precisely in (\ref{RonRoff}) and denoted by $\Ron$. The run time is the corresponding time spent attached to a microtubule, which was defined precisely in (\ref{TonToff}) and denoted by $\Ton$. The asymptotic velocity is
\begin{equation}
 	V := \lim_{t\to\infty}\frac{X(t)}{t} \label{eq:V}.
 \end{equation} 
 The velocity $V$ includes both the time the cargo is being transported along the microtubule and diffusing while unattached.
 
Applying Theorem~\ref{theorem renewal}, we have that
\begin{align}\label{Vr}
V=\frac{\E[\Ron]+\E[\Roff]}{\E[\Ton]+\E[\Toff]}\quad\text{almost surely}.
\end{align}
Now, 
\begin{align}\label{Roff}
\E[\Roff]=0,
\end{align}
since the cargo is freely diffusing when no motors are bound, and since motor binding and unbinding is independent of Brownian motion $\{W(t)\}_{t\ge0}$. Furthermore, when all of the $M$ motors are unbound, each motor binds at rate $\lon(0)$. Hence,
\begin{align}\label{Toff}
\E[\Toff]=(M\lon(0))^{-1}.
\end{align}
It therefore remains to calculate two of the three quantities, $V$, $\E[\Ron]$, and $\E[\Ton]$, since the third is given by (\ref{Vr}). We calculate $\E[\Ton]$ first since it is the simplest, as it is a mean first passage time of a continuous-time Markov chain.

\subsection{Expected run time}\label{section run time}

As we noted in section~\ref{section setup2}, the number of motors bound $m(t)$ is itself the Markov process (\ref{eq:nMarkov}). 
To compute the expected run time, we compute the mean time for $m(t)$ to reach state $m=0$ starting from $m(0)=1$.

Let $\widetilde{Q}\in\R^{(M+1)\times(M+1)}$ be the generator of the Markov chain $m(t)$ in (\ref{eq:nMarkov}). That is, the $(i,j)$-entry of $\widetilde{Q}$ gives the rate that $m(t)$ jumps from state $i$ to state $j\neq i$, and the diagonal entries are chosen so that $\widetilde{Q}$ has zero row sums. Let $Q\in\R^{M\times M}$ be the matrix obtained from deleting the first row and column of $\widetilde{Q}$. The matrix $Q$ is tridiagonal, with the $m$-th row containing  subdiagonal, diagonal, and superdiagonal entries $m\loff(m)$, $-(m\loff(m)+(M-m)\lon(m))$, $(M-m)\lon(m)$, respectively. The expected run time $\E[\Ton]$ is (by Theorem 3.3.3 in \cite{Norris1998}),
\begin{equation}
\E[\Ton] = \mathbf{1}^T \mathbf{t}, \qquad \text{where} \quad  Q^T \mathbf{t} = -\mathbf{e}_1\label{Ton},
\end{equation}
where $\textbf{1}\in\R^{M}$ is the vector of all $1$'s and $\mathbf{e}_1 \in \R^M$ is the standard basis vector.

\subsection{Decomposing stochasticity}

Having calculated $\E[\Ton]$ in (\ref{Ton}), we can determine $V$ by determining $\E[\Ron]$ (or vice versa). Two key steps allow us to analyze $V$ and $\E[\Ron]$: (i) we average over the diffusive dynamics while conditioning on a realization of the jump dynamics, and (ii) we take advantage of a timescale separation between the relaxation rate of the cargo dynamics and the jump rate of the motor dynamics.

\subsubsection{Conditioning on jump realizations} 
Observe that the stochasticity in the model can be separated into a \emph{continuous} diffusion part and a \emph{discrete} part controlling motor binding, unbinding, and stepping. Mathematically, the continuous diffusion part is described by the Brownian motion $W$ in (\ref{dim}), and the discrete motor state is described by the Markov jump process $\J$. We first average over the diffusion by defining the conditional expectations
\begin{align}\label{x}
\begin{split}
x(t)&:=\E[X(t)|\J],\quad
z_{i}(t):=\E[Z_{i}(t)|\J],\quad t\ge0,\;i\in\{1,\dots,M\}.
\end{split}
\end{align}
We emphasize that (\ref{x}) are averages over paths of $W$ given a realization $\J$. Thus, $\{x(t)\}_{t\ge0}$ and $\{\{z_{i}(t)\}_{t\ge0}\}_{i=1}^{M}$ are functions of the realization $\J$. This definition is convenient, because while $X(t)$ follows the randomly switching SDE (\ref{dim}), the process $x(t)$ follows a randomly switching ODE, whose solution is known explicitly.

\begin{prop}\label{propave}
For each $t>0$, the expected cargo position $x(t)$ conditioned on a realization of the jump process satisfies 
\begin{align}
\frac{\dd }{\dd t}x(t)
&=\eps^{-1}\sum_{i=1}^{M}\big(z_{i}(t)-x(t)\big),\quad\text{almost surely}\label{xx}.
\end{align}
\end{prop}

\begin{proof}
Using the explicit solution of an OU process, we have that
\begin{align}\label{ousol}
X(t)=X(\tau)e^{-\theta(t-\tau)}+\mu(1-e^{-\theta(t-\tau)})+\mathcal{M},
\end{align}
where $\tau$ is the most recent jump time of $\J$,
\begin{align*}
\tau=\sup\big\{\{0\}\cup\{s<t:\J(s-)\neq\J(s+)\}\big\},
\end{align*}
$\theta=m(\tau)\eps^{-1}$, $\mu=\frac{1}{m(\tau)}\sum_{i\in I(\tau)}Z_{i}(\tau)$, and $\mathcal{M}$ satisfies $\E[\mathcal{M}|\J]=0$. We have used the notation $f(t\pm):=\lim_{s\to t\pm}f(s)$. Hence, taking the expectation of (\ref{ousol}) conditioned on $\J$ yields
\begin{align*}
x(t)&=\E[X(\tau)e^{-\theta(t-\tau)}|\J]+\E[\mu(1-e^{-\theta(t-\tau)})|\J]\\
&=e^{-\theta(t-\tau)}\E[X(\tau)|\J]+(1-e^{-\theta(t-\tau)})\frac{1}{m(\tau)}\sum_{i\in I(\tau)}\E[Z_{i}(\tau)|\J],
\end{align*}
since $\tau$ are $\{m(s)\}_{s\ge0}$ measurable with respect to the $\sigma$-algebra generated by $\J$.
\end{proof}

\subsubsection{Separation of timescales}
We next make an observation of disparate timescales. After averaging over the diffusive noise $W$, the model effectively depends on two timescales: the relaxation time of the continuous dynamics (\ref{xx}) (characterized by the dimensionless rate $\eps^{-1}$)  and the switching times of the discrete motor dynamics (\ref{jumprates}) (characterized by the dimensionless rates $\lon,\lstep,\loff$). Even for conservative parameter estimates, the continuous timescale is much faster than the discrete switching timescale. For instance, suppose a motor exerts a Hookean force with stiffness $k=0.5$ pN/nm \cite{Furuta2013} on a spherical cargo with radius $r=1$ $\mu$m in cytosol with viscosity $\eta$ equal to that of water. It follows that $k/(6\pi\eta r)\approx 3\times10^{4}$ s${}^{-1}$, whereas $\koff(1)$ is on the order of $10^{-1}$ to $10^{1}$ s${}^{-1}$ \cite{Furuta2013}. Hence,
\begin{align}\label{kgamma}
\eps:=\koff(1)\gamma/k\approx3\textcolor{black}{\times}10^{-4}\ll1.
\end{align}
Further, $\lon,\lstep,\loff$ are roughly order one since $\kon,\kstep,\koff$ have similar orders of magnitude \cite{Furuta2013}.

Therefore, compared to the switching timescale, $x(t)$ quickly relaxes to an equilibrium between motor switches. Furthermore, we are interested in studying $\E[\Ron]$ and $V$, which depend on the behavior of $x(t)$ over the course of several motor switches. Hence, we approximate $x(t)$ by a jump process $\ol{x}(t)$ obtained from assuming $x(t)$ immediately relaxes to its equilibrium after each motor switch.

More precisely, let $(\x(t),\z_{1}(t),\dots,\z_{M}(t))\in\R^{M+1}$ be a $\J$-measurable, right-continuous process,
\begin{align*}
\x(t)=\x(t+),\quad\z_{i}(t)=\z_{i}(t+),\quad t\ge0,\;i\in\{1,\dots,M\},
\end{align*}
with $\x(0)=x(0)$ and $\z_{i}(0)=z_{i}(0)$, $i\in\{1,\dots,M\}$, that evolves in the following way. In light of (\ref{zp}), we define the effective motor positions by how they are modified through the jump process, binding at $\tau_i$ and then incrementing from stepping, or staying unbound at the cargo position $\x(t)$,
\begin{align}\label{zpb}
\z_{i}(t)=\big(\x(\tau_{i}(t))+J_{i}(t)\big)1_{\{J_{i}(t)\neq u\}}+\x(t)1_{\{J_{i}(t)=u\}},\quad i\in\{1,\dots,M\}.
\end{align}

Due to the assumed fast relaxation, $\x(t)$ only changes when a motor steps or unbinds, as newly bound motors exert no force. That is, if $J_{i}(t+)=J_{i}(t-)$ for all $i\in\{1,\dots,M\}$ satisfying $J_{i}(t-)\ge0$, then $\x(t-)=\x(t+)$. Otherwise, $\x(t)$ evolves according to the following two rules, which describe how the cargo position $\x(t)$ changes when a motor steps or unbinds.
\begin{enumerate}
\item
If the $i$-th motor steps ($J_{i}(t-)=j\ge0$ and $J_{i}(t+)=j+1$), then $\x(t+)=\x(t-)+1/m(t)$.
\item
If the $i$-th motor unbinds ($J_{i}(t-)=j\ge0$ and $J_{i}(t+)=u$), then $\x(t+)=\x(t-)+\Delta_{i,(\z_{1},\dots,\z_{m(t-)})}$, where $(\z_{1},\dots,\z_{m(t-)})$ gives the positions of the $m(t-)$ bound motors just before time $t$, and
\begin{align}\label{dj}
\Delta_{i,(\z_{1},\dots,\z_{m(t-)})}=\frac{1}{m(t-)-1}\sum_{i'=1,i'\neq i}^{m(t-)}\z_{i'}(t-)-\frac{1}{m(t-)}\sum_{i'=1}^{m(t-)}\z_{i'}(t-).
\end{align}
\end{enumerate}
In words, if either of these events occurs, the cargo position $\x(t)$ relaxes to the mean position of the motors. These two rules describe how the mean motor position changes in the two scenarios. If a single motor steps, incrementing its position by $1$, the mean motor position increases by $1/m(t)$. If a motor unbinds, \eqref{dj} is the change in the mean motor position from removing that motor.

It follows from these two evolution rules for $\overline{x}(t)$ that
\begin{align}\label{xoverg}
\x(t)=\sum_{m=1}^{M}\frac{1}{m}S_{m}(t)+\chi(t),
\end{align}
where $S_{m}(t)$ is the number of steps taken when $m$ motors are bound before time $t$ (each of which modifies the position by $1/m$), and $\chi(t)$ accounts for changes in the cargo position that result from a motor unbinding,
\begin{align}\label{deltag}
\chi(t)=\sum_{k=1}^{N_{\text{off}}(t)}\Delta_{j_{k},(\z_{1}(s_{\text{off}}^{k}-),\dots,\z_{m(s_{\text{off}}^{k}-)}(s_{\text{off}}^{k}-)},
\end{align}
where $0=s_{\text{off}}^{0}<s_{\text{off}}^{1}<\dots$ is the sequence of times in which a motor unbinds,
\begin{align*}
s_{\text{off}}^{k}:=\inf\big\{t>s_{\text{off}}^{k-1}:J_{i}(t-)\ge0\text{ and }J_{i}(t)=u\text{ for some }i\in\{1,\dots,M\}\big\}\quad k\ge1,
\end{align*}
and $N_{\text{off}}(t):=\sup\{k\ge0:s_{\text{off}}^{k}\le t\}$ is the number of unbindings before time $t\ge0$, 
and $j_{k}\in\{1,\dots,M\}$ gives the (almost surely unique) index of the motor that unbinds at time $s_{\text{off}}^{k}$. That is, $j_{k}$ satisfies $J_{j_{k}}(s_{\text{off}}^{k}-)
\neq J_{j_{k}}(s_{\text{off}}^{k})=u$.

The following proposition checks that $x(t)$ converges almost surely to the jump process $\x(t)$ as $\eps\to0$. The proof is in {Appendix \ref{sect proof prop conv}}. 
\begin{prop}\label{prop conv}
If $T\ge0$ is an almost surely finite stopping time with respect to $\{\J(t)\}_{t\ge0}$, then
\begin{align*}
\lim_{\eps\to0}x(T)=\x(T),\quad\text{almost surely}.
\end{align*}
\end{prop}
From this proposition, we conclude that studying the mean behavior of the cargo position $X(t)$ can ultimately be reduced to studying the jump process $\x(t)$, where the jumps correspond to motor stepping and unbinding events.

\subsection{Run length and velocity}

Since $\eps\ll1$ for biologically relevant parameters, we investigate the run length and velocity of $X(t)$ by investigating the analogous quantities for $\overline{x}(t)$,
\begin{align}\label{overq}
\overline{R}&:=\x(\toff^{1})-\x(\ton^{1}),
\qquad
\overline{V}:=\lim_{t\to\infty}\frac{\x(t)}{t}.
\end{align}
\subsubsection{Run length}
The following proposition checks that the mean run length of the full process $X(t)$ converges to the mean run length of the jump process $\x(t)$ as $\eps\to0$.
\begin{prop}\label{prop mean}
$\E[\Ron]\to\E[\overline{R}]$ as $\eps\to0$.
\end{prop}
\begin{proof}
By the tower property of conditional expectation (Theorem 5.1.6 in \cite{Durrett2014}), we have that
\begin{align*}
\E[\Ron]
=\E[\E[\Ron|\J]]
=\E[\E[X(\toff^{1})-X(\ton^{1})|\J]]
=\E[x(\toff^{1})-x(\ton^{1})].
\end{align*}
Now, Proposition~\ref{prop conv} ensures that
\begin{align}\label{cx}
x(\toff^{1})-x(\ton^{1})\to\overline{R},\quad\text{almost surely as }\eps\to0.
\end{align}
Let $N\ge0$ be the number of steps taken between time $\ton^{1}$ and time $\toff^{1}$. Since motors take steps of distance one, we have the almost sure bound, $x(\toff^{1})-x(\ton^{1})\le N$. Steps are taken at Poisson rate $m(t)\lstep(m(t))\le M\Lambda$, \textcolor{black}{where $\Lambda:=\max_{m\in\{1,\dots,M\}}\lstep(m)$}. Thus $\E[N]\le \Lambda M\E[\toff^{1}-\ton^{1}]<\infty$. Thus, (\ref{cx}) and the bounded convergence theorem complete the proof.
\end{proof}

\subsubsection{Velocity}

Let us now investigate $\overline{V}$ in (\ref{overq}), observing that this quantity can be approached in two ways. The first exploits the observation that non-zero mean displacements only occur from motor stepping, so the velocity can be interpreted as the product of how often a step occurs with $m$ motors and the size of the displacement. The second approach is again a reward-renewal argument, noting that the only non-zero displacements occur during epochs of bound cargo. The connection between these two approaches provides explicit relationships between the velocity, run lengths, and run times.  

Recalling the decomposition of the jump process $\x$  in (\ref{xoverg}), we seek to compute the expected value 
\begin{align*}
\E[\x(t)]=\sum_{m=1}^{m}\frac{1}{m}\E[S_{m}(t)]+\E[\chi(t)].
\end{align*}
Using the definition of $\chi(t)$ in (\ref{deltag}), we compute its expectation by summing over all possible displacements from one of $m$ motors unbinding $\Delta_{j,(\z_{1}(t-),\dots,\z_{m(s_{\text{off}}^{k}-)}(t-)}$, which yields
\begin{align*}
\sum_{j=1}^{m}\Delta_{j,(z_{1},\dots,z_{m})}=\frac{1}{m-1}\sum_{j=1}^{m}\Big(-z_{j}+\sum_{i=1}^{m}z_{j}\Big)-\sum_{i=1}^{m}z_{i}=0.
\end{align*}
Since each of the bound motors is equally likely to unbind, it follows that $\E[\chi(t)]=0$. This can be interpreted as the observation that the arithmetic mean does not change in expectation when removing a randomly (uniformly) chosen element. In other words, the effects of motors unbinding ahead of the cargo are completely offset in the mean by motors unbinding behind the cargo. Therefore, the only long-term influence on $\x$ is  stepping events. 

Given a realization $\{m(s)\}_{s\ge0}$, the number of steps taken with $m$ motors bound before time $t\ge0$ is Poisson distributed with mean $m\lstep(m)\int_{0}^{t}1_{m(s)=m}\,\dd s$. Hence,
\begin{align*}
\E[S_{m}(t)]=m\lstep(m)\E\Big[\int_{0}^{t}1_{m(s)=m}\,\dd s\Big].
\end{align*}
Now, $\{m(s)\}_{s\ge0}$ is an ergodic Markov process, so the occupation measure converges almost surely to the stationary measure (see Theorem 3.8.1 in \cite{Norris1998})
\begin{align*}
\frac{1}{t}\sum_{n=1}^{M}\int_{0}^{t}1_{m(s)=m}\,\dd s\to p_{m},\quad\text{almost surely as }t\to\infty,
\end{align*}
where $p_{m}:=\lim_{t\to\infty}\P(m(t)=m)$ is the stationary probability $m$ motors are bound. We note that $p_{m}$ is the $(m+1)$-st component of the unique probability vector, $\textbf{p}\in\R^{1\times(M+1)}$ satisfying (see Theorem 3.5.2 in \cite{Norris1998})
\begin{align}\label{pformula}
\textbf{p}\widetilde{Q}=0,
\end{align}
where $\widetilde{Q}\in\R^{(M+1)\times(M+1)}$ is the generator matrix defined in section~\ref{section run time}. Since the occupation measure is bounded above by one, the bounded convergence theorem gives
\begin{align*}
\lim_{t\to\infty}\frac{\E[\x(t)]}{t}
=\lim_{t\to\infty}\sum_{m=1}^{M}\lstep(m)\E\Big[\frac{1}{t}\int_{0}^{t}1_{m(s)=m}\,\dd s\Big]=\sum_{m=1}^{M}\lstep(m)p_{m}.
\end{align*}

It is easy to see that the classical renewal reward theorem applies to $\x(t)$ so that
\begin{align*}
\sum_{m=1}^{M}\lstep(m)p_{m}
=\lim_{t\to\infty}\frac{\E[\x(t)]}{t}
=\lim_{t\to\infty}\frac{\x(t)}{t}
=\frac{\E[\overline{R}]}{\E[\Ton]+\E[\Ton]},\quad\text{almost surely}.
\end{align*}
Furthermore, (\ref{Vr}), (\ref{Roff}), (\ref{Toff}), and Proposition~\ref{prop mean} yield
\begin{align*}
\lim_{\eps\to0}V
=\lim_{\eps\to0}\frac{\E[\Ron]}{(M\lon(0))^{-1}+\E[\Ton]}
=\frac{\E[\overline{R}]}{(M\lon(0))^{-1}+\E[\Ton]}
=\sum_{m=1}^{M}\lstep(m)p_{m}.
\end{align*}

In summary, we now have explicit formulas for the velocity $V$ and expected run length $\E[\Ron]$ of $X(t)$ in the small $\eps$ limit,
\begin{align}
\lim_{\eps\to0}V&=\overline{V}=\sum_{m=1}^{M}\lstep(m)p_{m}\label{totalV}\\
\lim_{\eps\to0}\E[\Ron]&=\E[\overline{R}]=\Big(\big(M\lon(0)\big)^{-1}+\E[\Ton]\Big)\sum_{m=1}^{M}\lstep(m)p_{m},\label{formulaR}
\end{align}
where $p_{m}$ is given by (\ref{pformula}) and $\E[\Ton]$ is given by (\ref{Ton}). In Fig.~\ref{fig:RV_fig}, we compare these formulas for $\E[\overline{R}]$ and $\overline{V}$ with estimates of $\E[\Ron]$ and $V$ from simulations of the full process $(X(t),\textbf{Z}(t),\J(t))$ (for details on our statistically exact simulation method, see section~\ref{section simulation}).

Furthermore, some experimental works \cite{Furuta2013,Jonsson2015} measure the average run velocity, $\E[\overline{R}/\Ton]$. Now, if $\sigma(m)$ denotes the $\sigma$-algebra generated by $\{m(t)\}_{t\ge0}$, then recalling (\ref{tauonoff}) and (\ref{TonToff}) and using the tower property of conditional expectation yields
\begin{align*}
\E\Big[\frac{\overline{R}}{\Ton}\Big]
=\E\Big[\frac{1}{\Ton}\E[\overline{R}|\sigma(m)]\Big]
&=\E\Big[\frac{1}{\Ton^{1}}\sum_{m=1}^{M}\frac{1}{m}\E[S_{m}(\toff^{1})|\sigma(m)]\Big]\\
&=\sum_{m=1}^{M}\frac{1}{m}m\lstep(m)\E\Big[\frac{1}{\Ton^{1}}\int_{0}^{\toff^{1}}1_{m(s)=m}\,\dd s\Big].
\end{align*}
Hence, it follows from (\ref{totalV}) that
\begin{align} \label{runV}
\VV := \E[\overline{R}/\Ton]=\overline{V}/\pon,
\end{align}
where $\pon=\sum_{m=1}^{M}p_{m}$ is the stationary probability that $m(t)\ge1$.

  \begin{figure}
 \centering
 \includegraphics[width=\textwidth]{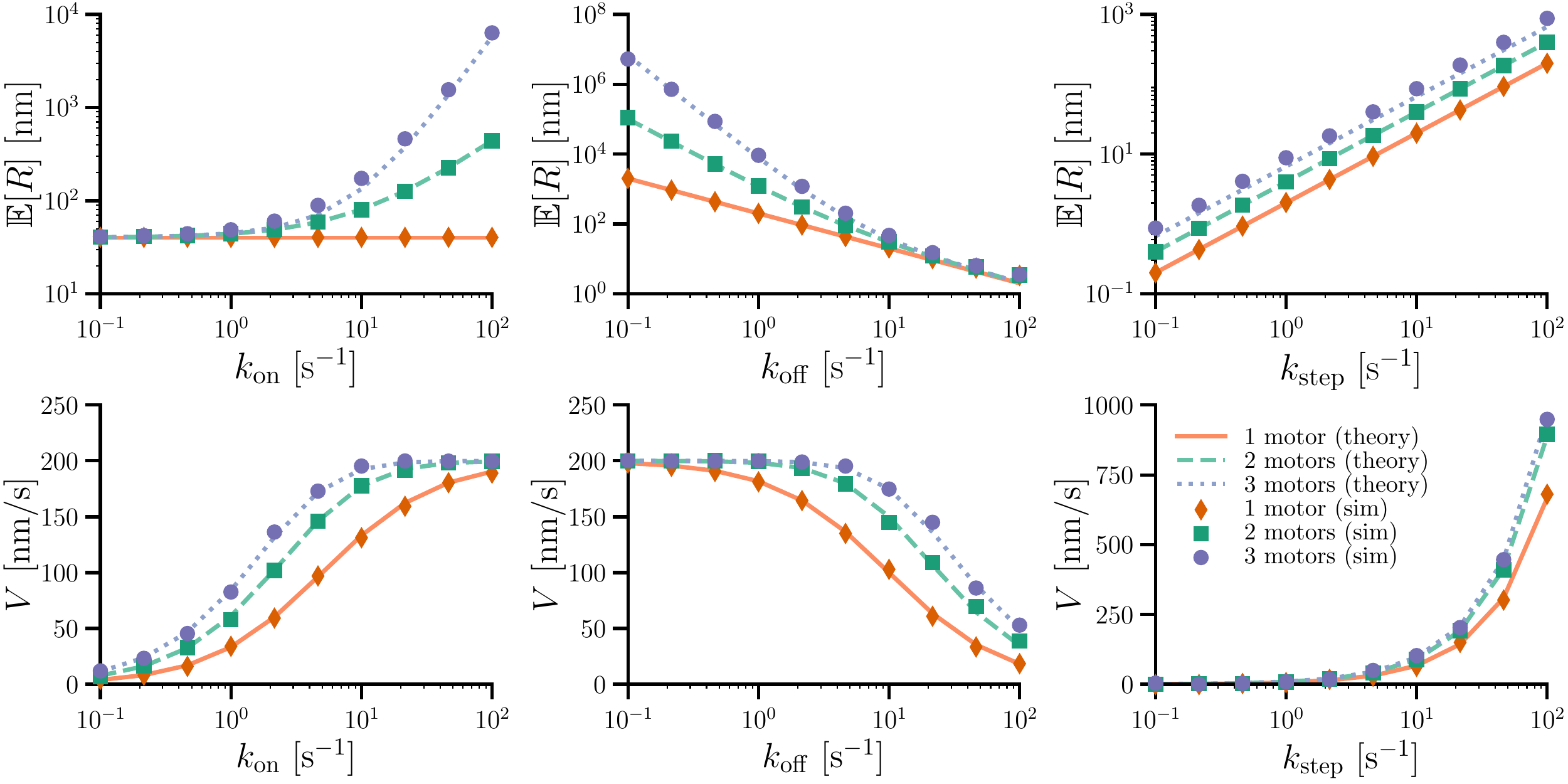}
 \caption{\label{fig:RV_fig}  Expected run lengths $\expect[R]$ and asymptotic velocities $V$ as a function of the parameters $\kon,\koff,\kstep$ for $M=1,2,3$ total motors. The curves are the analytical formulas (\ref{totalV})-(\ref{formulaR}) for the $\eps\to0$ limit, and the dots are estimates from statistically exact realizations of the full process, $\{(X(s),\textbf{Z}(s),\J(s))\}_{s=0}^{t}$, where the ending time $t$ is such that $N(t)=10^{5}$ where $N(t)$ is defined in (\ref{Ndefn}). Unless noted otherwise, $\kon(m) = 10 \left[\text{s}^{-1}\right], \kstep(m) = 20\left[\text{s}^{-1}\right], \koff(m) = 5 \left[\text{s}^{-1}\right]$ for each $m$. Further, $k$ and $\gamma$ are as in (\ref{kgamma}) and $k_{B}T=4.1\left[\text{pN}\cdot\text{nm}\right]$.
 }
 \end{figure}

\subsection{Cases $M=1$, $M=2$, and $M=3$}

In this subsection, we collect explicit formulas for the run length $\E[\overline{R}]$ and run velocity $\VV$ when the total number of motors is $M=1,2,3$. The run time $\E[\Ton]$ and net velocity $V$ can be easily deduced from these quantities using (\ref{totalV})-(\ref{formulaR}) but are omitted for brevity. 

For $M=1$ total motors, the quantities are simply
\begin{align}\label{M1}
\begin{split}
\E[\overline{R}]&=\lstep(1),\quad
\VV = \lstep(1).
\end{split}
\end{align}
For $M=2$ total motors, we find 
\begin{align}\label{M2}
\begin{split}
\E[\overline{R}]=\lstep(1)+ \frac{\lon(1)\lstep(2)}{2\loff(2)},\quad
\VV= \frac{2 \loff(2) \lstep(1)+\lon(1) \lstep(2)}{2 \loff(2)+\lon(1)}.
\end{split}
\end{align}
For $M=3$ total motors, we find {
\begin{align}\label{M3}
\begin{split}
\E[\overline{R}]&=\lstep(1)+\frac{\lon(1) (3 \loff(3) \lstep(2)+\lon(2) \lstep(3))}{3 \loff(2) \loff(3)},\\
\VV&=\frac{3 \loff(3) \left[\loff(2) \lstep(1)+\lon(1) \lstep(2)\right]+\lon(1) \lon(2) \lstep(3)}{3 \loff(3) \left[\loff(2)+\lon(1)\right]+\lon(1) \lon(2)}.
\end{split}
\end{align}
 }
To put these quantities in dimensional units, recall the jump rates (\ref{eq:rates}) and multiply $\E[\overline{R}]$ by the dimensional step distance $\delta>0$ and multiply $\VV$ by $\delta\koff(1)>0$.

\subsection{Numerical simulations}\label{section simulation}

To verify our predictions for the expected run lengths and velocities, we compare to statistically exact numerical simulations of the full process $(X(t),\textbf{Z}(t),\J(t))$. In a given state, we use the classical Gillespie stochastic simulation algorithm to generate the time of the next transition  for the Markov chain $\mathbf{J}(t)$ and to choose which transition occurs. For $m(t)\geq 1$, $X(t)$ is an OU process, generically described by 
\begin{equation*}
	\dd X(t) = \alpha \left[\mu - X(t)\right] \, \dd t + \beta \dd W(t).
\end{equation*}
To update $X(t)$ to the next time $t+\tau$, we use the statistically exact method described in \cite{Gillespie1996}, summarized by 
\begin{equation}
	X(t+\tau) = e^{-\alpha \tau} X(t) + (1-e^{-\alpha \tau}) \mu + \beta \sqrt{\frac{(1-e^{-2\alpha \tau})}{2\alpha}} \, \mathfrak{n} \label{eq:OU_update}
\end{equation}
where $\mathfrak{n}$ is a standard normal random variable. When $m(t)=0$, $X(t)$ is a pure diffusion process with $\alpha=0$, so \eqref{eq:OU_update} becomes an Euler-Maruyama update. This procedure generates statistically exact sample paths of $X(t)$, sampled at the transition times of $\mathbf{J}(t)$. We use this scheme to generate a long realization of $(X(t),\textbf{Z}(t),\J(t))$, thereby providing Monte Carlo estimates for $\E[\Ron]$ and $V$ for a given parameter set.

\section{Biological application}\label{section biology}

We now use the formulas (\ref{M1})-(\ref{M3}) for run velocity, $\VV$, and run length, $\E[\overline{R}]$, to explore the behavior of non-processive motors. The behavior of individual non-processive motors is characterized by two observations: i) short attachment times, ii) the time it takes to hydrolyze ATP (and consequently, to step) coincides with this attachment time \cite{Case1997,Foster2001,Nitzsche2016}. Concretely, Ncd motors in the kinesin-14 family take 1 to 5 steps before unbinding \cite{Astumian1999,Endow2000}. In our model, $\lstep(1)$ gives the expected number of steps before unbinding, so we characterize non-processive motors as those with $\lstep(1) \in [1,5]$.

 Using this characterization, we explore the observation made in \cite{Furuta2013,Jonsson2015,Yamada2017} that non-processive motors in the kinesin-14 family cooperate to produce long-range transport. This behavior is reported in \cite{Furuta2013,Jonsson2015} in terms of a velocity that is analogous to the run velocity $\VV$ in our model. Specifically, the primary manifestation of cooperativity is that $\VV$ increases substantially when the total number of motors increases from $M=1$ to $M=2$. For $M\geq 2$, the velocity remains relatively constant.
 
We thus ask the question: {what features are necessary to produce this behavior}?
Now, if the step rate is independent of the number of bound motors, $m$, then it follows immediately from (\ref{totalV}) and (\ref{runV}) that $\VV$ is independent of $M$. In particular, if the dimensional step rate is $\kstep(m)\equiv k_{0}$ for all $m\in\{1,\dots,M\}$, then in dimensional units, $\VV$ is simply $\delta k_{0}$, regardless of any other parameter values. 

Therefore, our model predicts that the stepping rate must depend on the number of bound motors in order to produce the cooperative behavior seen in run velocities in \cite{Furuta2013,Jonsson2015}. This prediction is bolstered by the simulation results of \cite{Furuta2013}. There, the authors constructed a detailed computational model of motor transport, and they had to improve motor stepping ability when two or more motors are bound in order for simulations of their computational model to match experimental run velocities. 

The authors of \cite{Furuta2013} also describe motor cooperativity in terms of the average distance traveled by a cargo before all of its motors detach from a microtubule, which is analogous to $\E[\overline{R}]$ in our model. Namely, they find that the run length $\E[\overline{R}]$ dramatically increases when $M$ increases from 1 to 2. Our model can replicate this cooperativity if and only if we allow the binding rate, $\kon$, and/or the unbinding rate, $\koff$, to depend on the number of bound motors, $m$. 

To illustrate, we find the parameter values needed for our model to match the measurements from \cite{Furuta2013}. However, we emphasize the qualitative results rather than the precise quantitative values of our parameters. Indeed, there are issues preventing an exact comparison of our model with the data in \cite{Furuta2013}. For example, as the authors point out, the length of the microtubules sometimes caused run lengths to be significantly altered (see Figures 1 and S6 in \cite{Furuta2013}). Furthermore, for a single motor ($M=1$), the authors report average run lengths of approximately 300 [nm], and they note that this value is necessarily an overestimate since they were unable to measure very short runs. Furthermore, this value must also be an overestimate since a single non-processive motor takes only a few steps per run (by definition of non-processive), and each step is approximately 7 [nm] \cite{Endow2000}.

We thus assume that $\lstep(1) = 4$, based on \cite{Astumian1999,Endow2000} and $\delta = 7 \, [\text{nm}]$. This gives $\E[R]=28 \, [\text{nm}]$ for $M=1$, which we use instead of the reported value in \cite{Furuta2013}. We then match the respective approximate run lengths of 1300 [nm] and 3300 [nm] for $M=1,2$ and the respective approximate run velocities of 100, 150, and 150 [nm/s] for $M=1,2,3$ reported in \cite{Furuta2013}. Using the formulas in (\ref{M1})-(\ref{M3}), this uniquely determines the stepping rates, $\kstep(1)\approx14$ [s$^{-1}$] and $\kstep(2)\approx\kstep(3)\approx21$ [s$^{-1}$], and the unbinding rate $\koff(1)\approx3.5$ [s$^{-1}$], which are all within the range of previously reported rates. The other binding/unbinding rates are not uniquely specified, but rather must satisfy the relations $\koff(2)\approx0.02\kon(1)$ and $\koff(3)\approx1.1\kon(2)$. Hence, if $\koff$ were constant in $m$, then $\kon(1)\approx200$ [s$^{-1}$] and $\kon(2)\approx3$ [s$^{-1}$].

We make two observations about this result: i) the binding rate $\kon(1)$ is an order of magnitude larger than reported values \cite{Astumian1999,Furuta2013} and ii) the binding rate decreases as the number of bound motors increases from 1 to 2. Both of these points can be explained by geometry. First, the value of $\kon(1)$ is enhanced because the single bound motor tethers the unbound motors close to the microtubule, and thus allows those motors to bind more easily. 
This binding enhancement due to geometry has precedent in motor studies. Indeed, in a different family of kinesins, it was shown to be critical for determining run lengths \cite{Feng2017}. Further, it was shown to play a critical role in enabling dynein processivity \cite{Grotjahn2017}, and it was posited as an explanation for why myosin motors can become processive when processive kinesin motors are present \cite{Hodges2009}. The authors in \cite{Braun2017} report large $\kon$ values in a model of  microtubule sliding driven by kinesin-14 motors and also speculate that this is due to tethering effects.

This effect can also be understood in terms of rebinding. If two motors are bound and one unbinds, then that motor can rapidly rebind since the bound motor keeps it near the microtubule. Such rebinding was the mechanism posited in \cite{Furuta2008} to explain the processive behavior of non-processive motors along microtubule bundles. Further, rebinding is very important in enzymatic reactions \cite{Takahashi2010,Gopich2013,Lawley2016,Lawley2017}. In that context, one incorporates rebinding by using an ``effective'' unbinding rate, which is the intrinsic unbinding rate multiplied by the probability that the particle does not rapidly rebind \cite{Lauffenburger1993}. Hence, this effect could be included in our model by reducing $\koff(2)$ rather than (or in addition to) increasing $\kon(1)$. Importantly, this is exactly what is implied by the relation, $\koff(2)\approx0.02\kon(1)$, derived above.

Second, geometric exclusion effects can explain a decrease in binding rates as the number of bound motors increases from 1 to 2. 
When more motors are bound, it is more difficult for additional motors to bind because the range of diffusive search is reduced for unbound motors. In numerical investigations of motor transport systems, this exact effect is observed \cite{Korn2009,Lombardo2017}. Furthermore this decrease in binding rate can arise due to motors competing for binding sites, a point posited in \cite{Klumpp2005}. Interestingly, these authors find that negative cooperativity has little impact on transport velocity. The same is true in our model, as the value of $\VV$ changes by less than 1 [nm/s] as $\kon(2)$ ranges from 0 to $\infty$ while keeping the other parameters fixed. However, we note that the run length for $M=3$ is greatly affected by $\kon(2)$, and thus this highlights the importance of using both run velocity and run length to study motor transport.
  
\section{Discussion}

In this work, we formulated and analyzed a mathematical model of transport by non-processive molecular motors. We deliberately made our model simple enough to enable us to extract explicit formulas for experimentally relevant quantities, yet maintain agreement with detailed computational studies.  One such simplification is to assume the motor stepping and unbinding rates are independent of force. The justification for this assumption is that since non-processive motors take only a few steps before unbinding (compared to hundreds of steps by processive motors), these motors are unlikely to be stretched long distances and therefore are unlikely to generate large forces. This assumption on the stepping rate has been made in other models involving non-processive motors  \cite{Nitzsche2016} and did not appear to be a necessary feature in that context, \textcolor{black}{and how force affects stepping is not completely clear \cite{Pechatnikova1999}. However, we note that force-dependent unbinding can be an important characteristic of processive motors, as kinesin-1 and kinesin-2 detach more rapidly under assisting load than under opposing load \cite{Arpag2014}, which increases run velocity and run length.}

These limitations notwithstanding, our model makes some concrete predictions about motor number-dependent stepping, binding, and unbinding behavior and how these quantities contribute to transport by non-processive motors. Specifically, we observe that a complex cooperativity mechanism appears to be a necessary ingredient for non-processive motor transport, and these predictions align with several recent experimental and computational works. Furthermore, these predictions can be investigated experimentally. Indeed, we hope that the work here will spur further investigation into how geometry affects non-processive motor transport, especially given that kinesin-14 motors are known to transport a wide variety of cargo, including long, cylindrical microtubules \cite{Hallen2008,Fink2009} and large, spherical vesicles in plants \cite{Yamada2017}.

\appendix

\section{Proofs of Lemmas~\ref{lemma off}-\ref{lemma finite}}\label{section lemma proofs}

\begin{proof}[Proof of Lemma~\ref{lemma off}]
Between time $\toff^{k-1}$ and $\ton^{k}$, the cargo is freely diffusing. Therefore, to control $\Yoff^{k}$, we need to control the supremum of a Brownian motion. Now, for any fixed $T>0$ and $C>0$, it follows from Doob's martingale inequality (Theorem 3.8(i) in \cite{Karatzas1991}) and symmetry of Brownian motion that
\begin{align}\label{bms}
\P(\sup_{t\in[0,T]}|W(t)|\ge C)\le 2\exp\Big(\frac{-C^{2}}{2T}\Big).
\end{align}
Hence, it follows that
\begin{align}\label{eb}
 \P(\Yoff^{k}\ge C|\Toff^{k})
 \le2\exp\Big(\frac{-C^{2}}{2\sigma^{2}\Toff^{k}}\Big),\quad\text{almost surely}.
\end{align}
Note that (\ref{eb}) is an average over realizations of the diffusion $W$ for fixed realizations of the time $\Toff^{k}$. That is, the inequality holds for almost all realizations of $\Toff^{k}$.

Now, $\Toff^{k}$ is exponentially distributed with rate $M\lon(0)$. Hence, the tower property of conditional expectation (see Theorem 5.1.6 in \cite{Durrett2014}) yields
\begin{align}\label{intoff}
 \P(\Yoff^{k}\ge C)
 =\E[\P(\Yoff^{k}\ge C|\Toff^{k})]
 &\le2M\lon(0)\int_{0}^{\infty}\exp\Big(\frac{-C^{2}}{2\sigma^{2}t}-M\lon(0)t\Big)\,\dd t.
\end{align}
Now, we have that
\begin{align}\label{genint}
\int_{0}^{\infty}\lambda e^{-\lambda t}e^{-a/t}\, \dd t
=2\sqrt{a\lambda}K_{1}(2\sqrt{a\lambda}),\quad\text{if }a>0,\,\lambda>0,
\end{align}
where $K_{1}(x)$ denotes the modified Bessel function of the second kind. Hence, the proof is complete after combining (\ref{intoff}) and (\ref{genint}) and the following bound, 
\begin{align*}
K_{1}(x)\le\sqrt{\pi/x}(1+1/(2x))e^{-x},\quad x>0,
\end{align*}
which was proven in \cite{Yang2017}.
\end{proof}

\begin{proof}[Proof of Lemma~\ref{lemma on}]
To control $\Yon^{k}$, we note that \textcolor{black}{if $t\in[\ton^{k},\toff^{k}]$, then $X(t)$ is an OU process centered at $\mu(t)=\frac{1}{m(t)}\sum_{i\in I(t)}Z_{i}(t)$ with relaxation rate $\theta(t)=\eps^{-1}m(t)$. Hence, after shifting time and space so that $\ton^{k}=0$ and $X(\ton^{k})=0$, we have that
\begin{align}\label{se12}
X(t)
=\int_{0}^{t}\mu(s)\theta(s)e^{-\int_{s}^{t}\theta(s')\,\dd s'}\,\dd s
+\sigma\int_{0}^{t}e^{-\int_{s}^{t}\theta(s')\,\dd s'}\,\dd W(s).
\end{align}
Since each bound motor takes steps of unit length at a Poisson rate, and since a motor binds at the current cargo location, it follows that
\begin{align}\label{se13}
|\mu(t)|
\le P(t)+X_{\text{sup}}(t),
\quad\text{where }X_{\text{sup}}(t):=\sup_{t'\in[0,t]} |X(t')|,
\end{align}
and $\{P(s)\}_{s\ge0}$ is a Poisson process with rate $\mu:=M\max_{m\in\{1,\dots,M\}}\lstep(m)$. Thus,
\begin{align}\label{ns3}
|X(t)|
\le
\sigma\Big|\int_{0}^{t}e^{-\int_{s}^{t}\theta(s')\,\dd s'}\,\dd W(s)\Big|
+P(t)
+\int_{0}^{t}X_{\text{sup}}(s)\theta(s)e^{-\int_{s}^{t}\theta(s')\,\dd s'}\,\dd s,
\end{align}
where we have used the fact that
\begin{align}\label{geu}
\int_{0}^{t_{0}}f(s)\theta(s)e^{-\int_{s}^{t_{0}}\theta(s')\,\dd s'}\,\dd s
\le
\int_{0}^{t}f(s)\theta(s)e^{-\int_{s}^{t}\theta(s')\,\dd s'}\,\dd s
\le f(t),\quad 0\le t_{0}\le t,
\end{align}
if $f$ is any nonnegative and nondecreasing function.}

\textcolor{black}{Now, a straightforward calculation using integration by parts shows that
\begin{align}\label{ibp}
\sup_{t'\in[0,t]}\Big|\int_{0}^{t'}e^{-\int_{s}^{t'}\theta(s')\,\dd s'}\,\dd W(s)\Big|
\le2\sup_{s\in[0,t]}|W(s)|.
\end{align}
Therefore, combining (\ref{ns3}), (\ref{geu}), and (\ref{ibp}) yields
\begin{align}\label{gp}
X_{\text{sup}}(t)
\le
\alpha(t)
+\int_{0}^{t}X_{\text{sup}}(s)\theta(s)e^{-\int_{s}^{t}\theta(s')\,\dd s'}\,\dd s,
\end{align}
where $\alpha(t):=2\sigma\sup_{s\in[0,t]}|W(s)|+P(t)$. Multiplying (\ref{gp}) by $e^{\int_{0}^{t}\theta(s')\,\dd s'}$, applying Gronwall's inequality, and then dividing by $e^{\int_{0}^{t}\theta(s')\,\dd s'}$ yields
\begin{align*}
X_{\text{sup}}(t)
\le
\alpha(t)+\int_{0}^{t}\alpha(s)\theta(s)\,\dd s,
\end{align*}
Since $\alpha(t)$ is nondecreasing and $\theta(t)\le\Theta:=\eps^{-1}M$, we obtain $X_{\text{sup}}(t)\le\alpha(t)(1+\Theta t)$. Therefore, we have the following almost sure inequality for $\zeta:=(1+\Theta \Ton^{k})^{-1}$,
\begin{align}
 \P(\Yon^{k}\ge C|\Ton^{k})
&\le\P\big(P(\Ton^{k})\ge C\zeta/2\big|\Ton^{k}\big)+\P\big(\sup_{t\in[0,\Ton^{k}]}|W(t)|\ge C\zeta/(4\sigma)\big|\Ton^{k}\big).\label{inc2}
\end{align}}

We thus need to control the distribution of $\Ton^{k}$. 
\textcolor{black}{Now, the Markov chain $m(t)$ in (\ref{eq:nMarkov})} is a finite state space birth-death process, and thus there exists \cite{Cavender1978} a unique quasi-stationary distribution $\nu\in\R^{M}$, which is a probability measure on $\{1,\dots,M\}$ so that if $\P(m(0)=m)=\nu_{m}$ for $m\in\{1,\dots,M\}$, then
\begin{align*}
\P\big(m(t)=m\big|m(s)\neq0\text{ for all }s\in[0,t]\big)=\nu_{m},\quad m\in\{1,\dots,M\}.
\end{align*}
Furthermore, it is known that the first passage time of $m(t)$ to state $0$ is exponentially distributed with some rate $\lambda>0$ if $\P(m(0)=m)=\nu_{m}$ for $m\in\{1,\dots,M\}$ \cite{Meleard2012}. \textcolor{black}{Now, $\Ton$ is the first passage time of $m(t)$ to state $m=0$ starting from state $m=1$, which must be less than the first passage time to state $m=0$ starting from any other state. Therefore, if $S$ is the first passage time to state $m=0$ starting from this quasi-stationary distribution, then}
\begin{align*}
\P(\Ton>T)\le\P(S>T)=1-e^{-\lambda T},\quad T>0.
\end{align*}
Thus, since both terms in the upper bound in (\ref{inc2}) are increasing functions of the realization $\Ton^{k}>0$, the tower property yields
\textcolor{black}{ for $\zeta_{S}:=(1+\Theta S)^{-1}$,
\begin{align}
\begin{split}\label{set4}
\P(\Yon^{k}\ge C\zeta)
&=\E\big[\P(\Yon^{k}\ge C\zeta|\Ton^{k})\big]\\
&\le\E\big[\P(P(S)\ge C\zeta_{S}/2|S)\big]+\E\big[\P(\sup_{t\in[0,S]}|W(t)|\ge C\zeta_{S}/(4\sigma)|S)\big].
\end{split}
\end{align}
}

Next, if $P$ is Poisson distributed with mean $\textcolor{black}{\mu_{0}}$, then Corollary 6 from \cite{short13} yields
\textcolor{black}{\begin{align*}
\P(P\ge C_{0})
\le e^{C_{0}-\mu_{0}}\Big(\frac{\mu_{0}}{C_{0}}\Big)^{C_{0}},\quad \text{if }C_{0}\ge\mu_{0}.
\end{align*}
Noting that $\frac{C}{2}< \mu S(1+\Theta S)$ if and only if $S>\Sigma:=\frac{\sqrt{2C\Theta/\mu+1}-1}{2\Theta}$, we obtain
\begin{align}\label{as3}
\begin{split}
&\P(P(S)\ge C\zeta_{S}/2|S)
\le
1_{S>\Sigma}+e^{\frac{C}{2(1+\Theta S)}-\mu S}\Big(\frac{2\mu S(1+\Theta S)}{C}\Big)^{\frac{C}{2(1+\Theta S)}}1_{S\le\Sigma}.
\end{split}
\end{align}}
\textcolor{black}{Since $S\sim\text{Exponential}(\lambda)$, taking the expectation of (\ref{as3}) gives 
\begin{align}\label{eg67}
E[\P(P(S)\ge C\zeta_{S}/2|S)]
\le
e^{-\lambda\Sigma}
+\int_{0}^{\Sigma}
\lambda e^{\frac{C}{2(1+\Theta s)}-(\mu+\lambda) s}\Big(\frac{2\mu s(1+\Theta s)}{C}\Big)^{\frac{C}{2(1+\Theta s)}}\,\dd s.
\end{align}
A quick calculation shows that if $C\ge(\Theta\eps_{0})^{-2}$ for $\eps_{0}:=(8e\mu\Theta)^{-1/2}$, then 
\begin{align*}
\lambda e^{\frac{C}{2(1+\Theta s)}-(\mu+\lambda) s}\Big(\frac{2\mu s(1+\Theta s)}{C}\Big)^{\frac{C}{2(1+\Theta s)}}
\le
\begin{cases}
\lambda2^{-\kappa_{0}\sqrt{C}} & \text{if }s\in[0,\eps_{0}\sqrt{C}],\\
\lambda e^{-\lambda \eps_{0}\sqrt{C}} & \text{if }s\in[\eps_{0}\sqrt{C},\Sigma],
\end{cases}
\end{align*}
for $\kappa_{0}:=(4\Theta\eps_{0})^{-1}$. Hence, if $C\ge(\Theta\eps_{0})^{-2}$, then
\begin{align}\label{bound111}
E[\P(P(S)\ge C\zeta_{S}/2|S)]
\le\lambda e^{-\lambda\Sigma}+\Sigma\lambda(2^{-\kappa_{0}\sqrt{C}}+e^{-\lambda \eps_{0}\sqrt{C}}).
\end{align}
}

\textcolor{black}{Moving to the second term in (\ref{set4}), we have that (\ref{bms}) yields
\begin{align*}
\E\big[\P(\sup_{t\in[0,S]}|W(t)|\ge C\zeta_{S}/(4\sigma)|S)\big]
\le\int_{0}^{\infty}2\lambda e^{-\lambda s}\exp\Big\{-\frac{\kappa C^{2}}{s(1+\Theta s)^{2}}\Big\}\,\dd s,
\end{align*}
where $\kappa=(32\sigma^{2})^{-1}$. It is straightforward to check that
\begin{align*}
2\lambda e^{-\lambda s}\exp\Big\{-\frac{\kappa C^{2}}{s(1+\Theta s)^{2}}\Big\}
\le
\begin{cases}
2\lambda\exp\Big\{-\frac{\kappa \sqrt{C}}{\Theta^{2}+2\Theta C^{-1/2}+C^{-1}}\Big\} & \text{if }s\in[0,\sqrt{C}],\\
2\lambda e^{-\lambda s} & \text{if }s\in[\sqrt{C},\infty),
\end{cases}
\end{align*}
Therefore, for sufficiently large $C$, we have that
\begin{align}\label{bound222}
\E\big[\P(\sup_{t\in[0,S]}|W(t)|\ge C\zeta_{S}/(4\sigma)|S)\big]
\le
2\lambda\sqrt{C} e^{-(2\kappa/\theta^{2})\sqrt{C}}+2e^{-\lambda\sqrt{C}}
\end{align}
Combining (\ref{set4}), (\ref{bound111}), and (\ref{bound222}) completes the proof.}
\end{proof}

\begin{proof}[Proof of Lemma~\ref{lemma bc}]
 \textcolor{black}{Since $Y_{k}\le\Yoff^{k}+\Yon^{k}$ for $k\ge1$, we have that
 \begin{align}\label{series3}
 \sum_{k=1}^{\infty}\P(Y_{k}>\sqrt{k})
 &\le\sum_{k=1}^{\infty}\P(\Yoff^{k}>\tfrac{1}{2}\sqrt{k})
+ \sum_{k=1}^{\infty}\P(\Yon^{k}>\tfrac{1}{2}\sqrt{k}).
 \end{align}
Therefore, the upper bounds established in Lemmas~\ref{lemma off} and \ref{lemma on} and the integral test show that (\ref{series3}) converges. Applying the Borel-Cantelli lemma (Theorem 2.3.1 in \cite{Durrett2014}) completes the proof.}
\end{proof}

\begin{proof}[Proof of Lemma~\ref{lemma finite}]
Since $Y_{k}\ge0$ almost surely, we have that $\E[Y_{k}]=\int_{0}^{\infty}\P(Y_{k}>C)\,\mathrm{d}C$. Using the bounds in Lemmas~\ref{lemma off} and \ref{lemma on} as in the proof of Lemma~\ref{lemma bc} shows that this integral is finite.
\end{proof}

\section{Proof of Proposition \ref{prop conv}} \label{sect proof prop conv}
\begin{proof}
Fix a realization $\J$. Let $K\ge0$ denote the almost surely finite number of jump times of $\J$ before time $T$, where $t$ is said to be a jump time if $\J(t+)\neq \J(t-)$. Denote these $K$ jump times by $0<\tau_{1}<\dots<\tau_{K}<T$ and let $\tau_{0}=0$ and $\tau_{K+1}=T$. 

For ease of notation, define the sequences
\begin{align*}
x_{k}:=x(\tau_{k}),\quad
\x_{k}:=\x(\tau_{k}),\quad
z_{k}^{i}:=z_{i}(\tau_{k}),\quad
\z_{k}^{i}:=\z_{i}(\tau_{k}),\quad
m_{k}:=m(\tau_{k}),
\end{align*}
for $k\in\{0,1,\dots,K\}$. Further, define the time between jumps, $s_{k}:=\tau_{k}-\tau_{k-1}$, for $k\in\{1,\dots,K\}$. It follows immediately from Proposition~\ref{propave} that
\begin{align}\label{it}
x_{k+1}=x_{k}e^{-s_{k+1}/\eps}+\mu_{k+1}(1-e^{-s_{k+1}/\eps}),\quad k\in\{0,1,\dots,K\},
\end{align}
where for $k\in\{0,1,\dots,K+1\}$ we define
\begin{align}\label{mu}
\mu_{k+1}
:=\begin{cases}
\frac{1}{m_{k}}\sum_{i\in I(\tau_{k})}z_{k}^{i} & \text{if }m_{k}>0,\\
x_{k} & \text{if }m_{k}=0.
\end{cases}
\end{align}
Furthermore, it follows from the definition of $\x(t)$ that for $k\in\{0,1,\dots,K\}$,
\begin{align}\label{xb}
\x_{k+1}
:=\begin{cases}
\frac{1}{m_{k}}\sum_{i\in I(\tau_{k})}\z_{k}^{i} & \text{if }m_{k}>0,\\
\x_{k} & \text{if }m_{k}=0.
\end{cases}
\end{align}

Now, since motors take steps of size one, it follows that if $k\in\{0,\dots,K\}$ and $i\in\{1,\dots,M\}$, then $0\le z_{k}^{i}\le K+1$ and $0\le x_{k}\le K+1$. Hence, if $k\in\{0,\dots,K\}$, then (\ref{mu}) implies
\begin{align}\label{bk}
|x_{k}-\mu_{k+1}|<K+1.
\end{align}

Next, we claim that if $k\in\{0,\dots,K\}$ and 
\begin{align}\label{c1}
\max_{j\in\{0,\dots,k\}}\Big\{|x_{j}-\x_{j}|,\max_{i\in\{1,\dots,M\}}|z_{j}^{i}-\z_{j}^{i}|\Big\}<\eta,
\end{align}
then
\begin{align}\label{c2}
\max\Big\{|x_{k+1}-\x_{k+1}|,\max_{i\in\{1,\dots,M\}}|z_{k+1}^{i}-\z_{k+1}^{i}|\Big\}<(K+1)e^{-s_{k+1}/\eps}+\eta.
\end{align}
To see this, we use (\ref{it}) and (\ref{bk}) to obtain
\begin{align*}
|x_{k+1}-\x_{k+1}|
&=|x_{k}e^{-s_{k+1}/\eps}+\mu_{k+1}(1-e^{-s_{k+1}/\eps})-\x_{k+1}|\\
&\le(K+1)e^{-s_{k+1}/\eps} + |\mu_{k+1}-\x_{k+1}|.
\end{align*}
Using (\ref{mu}) and (\ref{xb}), we have that
\begin{align*}
|\mu_{k+1}-\x_{k+1}|
&\le\begin{cases}
\frac{1}{m_{k}}\sum_{i\in I(\tau_{k})}|z_{k}^{i}-\z_{k}^{i}| & \text{if }m_{k}>0,\\
|x_{k}-\x_{k}| & \text{if }m_{k}=0.
\end{cases}
\end{align*}
Furthermore, it follows from (\ref{zp}) and (\ref{zpb}) that 
\begin{align*}
|z_{k+1}^{i}-\z_{k+1}^{i}|\le \max_{j\in\{0,\dots,K+1\}}|x_{j}-\x_{j}|,\quad i\in\{1,\dots,M\}.
\end{align*}
Hence, the claim (\ref{c2}) is verified.

Define the largest time between jumps, $s:=\max_{k\in\{1,\dots,K\}}s_{k}$. Since $x_{0}=\x_{0}=z_{0}^{i}=\z_{0}^{i}$ for $i\in\{1,\dots,M\}$, we apply (\ref{c1}) and (\ref{c2}) iteratively to obtain
\begin{align*}
|x_{K+1}-\x_{K+1}|\le(K+1)^{2}e^{-s/\eps}.
\end{align*}
Taking $\eps\to0$ completes the proof.
\end{proof}

\bibliography{library}

\begin{thebibliography}{10}

\bibitem{Arpag2014}
{\sc G.~Arpag, S.~Shastry, W.~O. Hancock, and E.~T{\"{u}}zel}, {\em {Transport
  by Populations of Fast and Slow Kinesins Uncovers Novel Family-Dependent
  Motor Characteristics Important for In Vivo Function}}, Biophys. J., 107
  (2014), pp.~1896--1904.

\bibitem{Astumian1999}
{\sc R.~D. Astumian and I.~Der{\'{e}}nyi}, {\em {A chemically reversible
  Brownian motor: application to kinesin and Ncd.}}, Biophys. J., 77 (1999),
  pp.~993--1002.

\bibitem{Braun2017}
{\sc M.~Braun, Z.~Lansky, A.~Szuba, F.~W. Schwarz, A.~Mitra, M.~Gao,
  A.~L{\"{u}}decke, P.~R. ten Wolde, and S.~Diez}, {\em {Changes in microtubule
  overlap length regulate kinesin-14-driven microtubule sliding}}, Nat. Chem.
  Biol.,  (2017).

\bibitem{Bressloff2014}
{\sc P.~C. Bressloff}, {\em {Stochastic Processes in Cell Biology}}, vol.~41 of
  Interdisciplinary Applied Mathematics, Springer International Publishing,
  2014.

\bibitem{Bressloff2013}
{\sc P.~C. Bressloff and J.~M. Newby}, {\em {Stochastic models of intracellular
  transport}}, Rev. Mod. Phys., 85 (2013), pp.~135--196.

\bibitem{Case1997}
{\sc R.~B. Case, D.~W. Pierce, N.~Hom-Booher, C.~L. Hart, and R.~D. Vale}, {\em
  {The directional preference of kinesin motors is specified by an element
  outside of the motor catalytic domain}}, Cell, 90 (1997), pp.~959--966.

\bibitem{Cavender1978}
{\sc J.~A. Cavender}, {\em Quasi-stationary distributions of birth-and-death
  processes}, Advances in Applied Probability, 10 (1978), pp.~570--586.

\bibitem{Davis1984}
{\sc M.~H.~A. Davis}, {\em Piecewise-deterministic markov processes: A general
  class of non-diffusion stochastic models}, J R Stat Soc Series B Stat
  Methodol.,  (1984), pp.~353--388.

\bibitem{Duke2000a}
{\sc T.~Duke}, {\em {Cooperativity of myosin molecules through strain-dependent
  chemistry.}}, Philos. Trans. R. Soc. Lond. B. Biol. Sci., 355 (2000),
  pp.~529--38.

\bibitem{Durrett2014}
{\sc R.~Durrett}, {\em Probability: theory and examples}, Cambridge university
  press, 2010.

\bibitem{Endow2000}
{\sc S.~A. Endow and H.~Higuchi}, {\em {A mutant of the motor protein kinesin
  that moves in both directions on microtubules.}}, Nature, 406 (2000),
  pp.~913--6.

\bibitem{Erickson2011}
{\sc R.~P. Erickson, Z.~Jia, S.~P. Gross, and C.~C. Yu}, {\em {How molecular
  motors are arranged on a cargo is important for vesicular transport}}, PLoS
  Comput. Biol., 7 (2011).

\bibitem{Feng2017}
{\sc Q.~Feng, K.~J. Mickolajczyk, G.-Y. Chen, and W.~O. Hancock}, {\em {Motor
  reattachment kinetics play a dominant role in multimotor-driven cargo
  transport}}, Biophys. J., 114 (2017), pp.~1--12.

\bibitem{Fink2009}
{\sc G.~Fink, L.~Hajdo, K.~J. Skowronek, C.~Reuther, A.~A. Kasprzak, and
  S.~Diez}, {\em {The mitotic kinesin-14 Ncd drives directional
  microtubule-microtubule sliding.}}, Nat. Cell Biol., 11 (2009), pp.~717--23.

\bibitem{Foster2001}
{\sc K.~A. Foster, A.~T. Mackey, and S.~P. Gilbert}, {\em {A Mechanistic Model
  for Ncd Directionality}}, J. Biol. Chem., 276 (2001), pp.~19259--19266.

\bibitem{Furuta2013}
{\sc K.~Furuta, A.~Furuta, Y.~Y. Toyoshima, M.~Amino, K.~Oiwa, and H.~Kojima},
  {\em {Measuring collective transport by defined numbers of processive and
  nonprocessive kinesin motors}}, Proc. Natl. Acad. Sci., 110 (2013),
  pp.~501--506.

\bibitem{Furuta2008}
{\sc K.~Furuta and Y.~Y. Toyoshima}, {\em {Minus-End-Directed Motor Ncd
  Exhibits Processive Movement that Is Enhanced by Microtubule Bundling In
  Vitro}}, Curr. Biol., 18 (2008), pp.~152--157.

\bibitem{Gillespie1996}
{\sc D.~T. Gillespie}, {\em {Exact numerical simulation of the
  Ornstein-Uhlenbeck process and its integral}}, Phys. Rev. E, 54 (1996),
  pp.~2084--2091.

\bibitem{Gopich2013}
{\sc I.~V. Gopich and A.~Szabo}, {\em Diffusion modifies the connectivity of
  kinetic schemes for multisite binding and catalysis}, Proc. Natl. Acad. Sci.
  USA, 110 (2013), pp.~19784--19789.

\bibitem{Grotjahn2017}
{\sc D.~A. Grotjahn, S.~Chowdhury, Y.~Xu, R.~J. McKenney, T.~A. Schroer, and
  G.~C. Lander}, {\em {Cryo-electron tomography reveals that dynactin recruits
  a team of dyneins for processive motility}}, Nat. Struct. Mol. Biol., 25
  (2018), pp.~203--207.

\bibitem{Hallen2008}
{\sc M.~A. Hallen, Z.-Y. Liang, and S.~A. Endow}, {\em {Ncd motor binding and
  transport in the spindle.}}, J. Cell Sci., 121 (2008), pp.~3834--41.

\bibitem{Hodges2009}
{\sc A.~R. Hodges, C.~S. Bookwalter, E.~B. Krementsova, and K.~M. Trybus}, {\em
  {A Nonprocessive Class V Myosin Drives Cargo Processively When a Kinesin-
  Related Protein Is a Passenger}}, Curr. Biol., 19 (2009), pp.~2121--2125.

\bibitem{Hughes2011}
{\sc J.~Hughes, W.~O. Hancock, and J.~Fricks}, {\em {A matrix computational
  approach to kinesin neck linker extension}}, J. Theor. Biol., 269 (2011),
  pp.~181--194.

\bibitem{Hughes2012}
\leavevmode\vrule height 2pt depth -1.6pt width 23pt, {\em {Kinesins with
  Extended Neck Linkers: A Chemomechanical Model for Variable-Length
  Stepping}}, Bull. Math. Biol., 74 (2012), pp.~1066--1097.

\bibitem{Hughes2013}
{\sc J.~Hughes, S.~Shastry, W.~O. Hancock, and J.~Fricks}, {\em {Estimating
  Velocity for Processive Motor Proteins with Random Detachment}}, J. Agric.
  Biol. Environ. Stat., 18 (2013), pp.~204--217.

\bibitem{Jonsson2015}
{\sc E.~Jonsson, M.~Yamada, R.~D. Vale, and G.~Goshima}, {\em {Clustering of a
  kinesin-14 motor enables processive retrograde microtubule-based transport in
  plants}}, Nat. Plants, 1 (2015), p.~15087.

\bibitem{Julicher1995}
{\sc F.~J{\"{u}}licher and J.~Prost}, {\em {Cooperative molecular motors}},
  Phys. Rev. Lett., 75 (1995), pp.~2618--2621.

\bibitem{Karatzas1991}
{\sc I.~Karatzas and S.~Shreve}, {\em Brownian motion and stochastic calculus},
  vol.~113, Springer Science \& Business Media, 2012.

\bibitem{Klumpp2005}
{\sc S.~Klumpp and R.~Lipowsky}, {\em {Cooperative cargo transport by several
  molecular motors}}, Proc. Natl. Acad. Sci., 102 (2005), pp.~17284--17289.

\bibitem{Korn2009}
{\sc C.~B. Korn, S.~Klumpp, R.~Lipowsky, and U.~S. Schwarz}, {\em {Stochastic
  simulations of cargo transport by processive molecular motors}}, J. Chem.
  Phys., 131 (2009).

\bibitem{Krishnan2011}
{\sc A.~Krishnan and B.~I. Epureanu}, {\em {Renewal-Reward Process Formulation
  of Motor Protein Dynamics}}, Bull. Math. Biol., 73 (2011), pp.~2452--2482.

\bibitem{Lauffenburger1993}
{\sc D.~A. Lauffenburger and J.~Linderman}, {\em Receptors: models for binding,
  trafficking, and signaling}, Oxford University Press, 1993.

\bibitem{Lawley2016}
{\sc S.~D. Lawley and J.~P. Keener}, {\em {Including Rebinding Reactions in
  Well-Mixed Models of Distributive Biochemical Reactions}}, Biophys. J., 111
  (2016), pp.~2317--2326.

\bibitem{Lawley2017}
\leavevmode\vrule height 2pt depth -1.6pt width 23pt, {\em {Rebinding in
  biochemical reactions on membranes}}, Phys. Biol.,  (2017).

\bibitem{Lombardo2017}
{\sc A.~T. Lombardo, S.~R. Nelson, M.~Y. Ali, G.~G. Kennedy, K.~M. Trybus,
  S.~Walcott, and D.~M. Warshaw}, {\em {Myosin Va molecular motors manoeuvre
  liposome cargo through suspended actin filament intersections in vitro.}},
  Nat. Commun., 8 (2017), p.~15692.

\bibitem{McKinley2011a}
{\sc S.~A. McKinley, A.~Athreya, J.~Fricks, and P.~R. Kramer}, {\em {Asymptotic
  analysis of microtubule-based transport by multiple identical molecular
  motors}}, J. Theor. Biol., 305 (2012), pp.~54--69.

\bibitem{Meleard2012}
{\sc S.~M{\'e}l{\'e}ard, D.~Villemonais, et~al.}, {\em Quasi-stationary
  distributions and population processes}, Probability Surveys, 9 (2012),
  pp.~340--410.

\bibitem{Mieck2015}
{\sc C.~Mieck, M.~I. Molodtsov, K.~Drzewicka, B.~van~der Vaart, G.~Litos,
  G.~Schmauss, A.~Vaziri, and S.~Westermann}, {\em {Non-catalytic motor domains
  enable processive movement and functional diversification of the kinesin-14
  Kar3}}, Elife, 4 (2015), pp.~1--23.

\bibitem{Nitzsche2016}
{\sc B.~Nitzsche, E.~Dudek, L.~Hajdo, A.~A. Kasprzak, A.~Vilfan, and S.~Diez},
  {\em {Working stroke of the kinesin-14, ncd, comprises two substeps of
  different direction}}, Proc. Natl. Acad. Sci., 113 (2016), pp.~E6582--E6589.

\bibitem{Norris1998}
{\sc J.~Norris}, {\em {Markov Chains}}, Statistical {\&} Probabilistic
  Mathematics, Cambridge University Press, 1998.

\bibitem{Pechatnikova1999}
{\sc E.~Pechatnikova and E.~W. Taylor}, {\em {Kinetics processivity and the
  direction of motion of Ncd.}}, Biophys. J., 77 (1999), pp.~1003--1016.

\bibitem{Shaklee2008}
{\sc P.~M. Shaklee, T.~Idema, G.~Koster, C.~Storm, T.~Schmidt, and
  M.~Dogterom}, {\em {Bidirectional membrane tube dynamics driven by
  nonprocessive motors}}, Proc. Natl. Acad. Sci., 105 (2008), pp.~7993--7997.

\bibitem{short13}
{\sc M.~Short}, {\em Improved inequalities for the poisson and binomial
  distribution and upper tail quantile functions}, ISRN Probability and
  Statistics, 2013 (2013).

\bibitem{Shtylla2015}
{\sc B.~Shtylla and J.~P. Keener}, {\em {Mathematical modeling of bacterial
  track-altering motors : Track cleaving through burnt-bridge ratchets}}, Phys.
  Rev. E, 91 (2015), p.~042711.

\bibitem{Shubeita2008}
{\sc G.~T. Shubeita, S.~L. Tran, J.~Xu, M.~Vershinin, S.~Cermelli, S.~L.
  Cotton, M.~A. Welte, and S.~P. Gross}, {\em {Consequences of Motor Copy
  Number on the Intracellular Transport of Kinesin-1-Driven Lipid Droplets}},
  Cell, 135 (2008), pp.~1098--1107.

\bibitem{Takahashi2010}
{\sc K.~Takahashi, S.~Tanase-Nicola, and P.~Rein Ten~Wolde}, {\em
  Spatio-temporal correlations can drastically change the response of a {MAPK}
  pathway}, Proc. Natl. Acad. Sci. USA, 107 (2010), pp.~2473--2478.

\bibitem{Vlasiou2011}
{\sc M.~Vlasiou}, {\em Regenerative processes}, Wiley Encyclopedia of
  Operations Research and Management Science,  (2011).

\bibitem{Yamada2017}
{\sc M.~Yamada, Y.~Tanaka-Takiguchi, M.~Hayashi, M.~Nishina, and G.~Goshima},
  {\em {Multiple kinesin-14 family members drive microtubule minus end-directed
  transport in plant cells}}, J. Cell Biol., 216 (2017), pp.~1705--1714.

\bibitem{Yang2017}
{\sc Z.-H. Yang and Y.-M. Chu}, {\em On approximating the modified bessel
  function of the second kind}, Journal of Inequalities and Applications, 2017
  (2017), p.~41.

\end{thebibliography}
\bibliographystyle{siam}
\end{document}